\documentclass[10pt,a4paper]{article}
\usepackage{latexsym,amsthm,amssymb,amscd,amsmath}
\newcounter{alphthm}
\theoremstyle{plain}
\newtheorem{theorem}{Theorem}[section]
\newtheorem{lemma}[theorem]{Lemma}
\newtheorem{proposition}[theorem]{Proposition}
\newtheorem{cor}[theorem]{Corollary}
\theoremstyle{definition}

\newcommand{\be}{\begin{equation}}
\newcommand{\ee}{\end{equation}}
\newcommand{\ben}{\begin{enumerate}}
\newcommand{\een}{\end{enumerate}}

\begin{document}

\title{ Looking for K\"{a}hler- Einstein Structure on Cartan Spaces with Berwald connection}
\author{E. Peyghan, A. Tayebi and A. Ahmadi}

\maketitle
\begin{abstract}
A Cartan manifold is a smooth manifold $M$ whose slit cotangent
bundle $T^{\ast}M_{0}$ is endowed with a regular Hamiltonian $K$
which is positively homogeneous of degree 2 in momenta. The
Hamiltonian $K$ defines a (pseudo)-Riemannian metric $g_{ij}$ in the
vertical bundle over $T^{\ast}M_{0}$ and using it a Sasaki type
metric on $T^{\ast}M_{0}$ is constructed. A natural almost complex
structure is also defined by $K$ on $T^{\ast}M_{0}$ in such a way
that pairing it with the Sasaki type metric an almost K\"{a}hler
structure is obtained. In this paper we deform $g_{ij}$ to a
pseudo-Riemannian metric $G_{ij}$ and we define a corresponding
almost complex K\"{a}hler structure. We determine the Levi-Civita
connection of $G$ and compute all the components of its curvature.
Then we prove that if the structure $(T^{\ast}M_{0},G,J)$ is
K\"{a}hler- Einstein, then the Cartan structure given by $K$ reduce
to a Riemannian one.
\end{abstract}

Laplace operator, Divergence, Gradient.

\section{Introduction}
\'{E}. Cartan  has originally introduced a Cartan space, which is considered as dual of Finsler space \cite{Car}. Then  H. Rund \cite{Rund}, F. Brickell \cite{Bric} and others studied the
relation between these two spaces. The theory of Hamilton
spaces was introduced and studied by R. Miron (\cite{Miro1}, \cite{Miro2}).
He proved that Cartan space is a particular
case of Hamilton space. Indeed the geometry of regular Hamiltonians
as smooth functions on the cotangent bundle is due to R. Miron and it
is now systematically described in the monograph \cite{MirHri}.

Let us denote the Hamiltonian structure on a manifold $M$ by $(M, H(x,p))$. If the fundamental function $H(x,p)$ is 2-homogeneous on the fibres of the cotangent bundle $(T^*M, M)$, then the
notion of Cartan space is obtained. The modern formulation of the notion of Cartan spaces is due of
the R. Miron \cite{Mir1}, \cite{Mir2}, \cite{Mir3}.
 Based on the studies of E. Cartan, A. Kawaguchi \cite{Kaw}, R.
Miron \cite{MirHri}, \cite{Mir2}, \cite{Mir3}, S. Vacaru \cite{V0,V1,V2},  D. Hrimiuc and
H. Shimada \cite{Hri}, \cite{HriShi},
 P.L. Antonelli and  M. Anastasiei  \cite{Anas}, \cite{Anast} \cite{MirAn}, \cite{M2}, etc.,
 the geometry of Cartan spaces is today an important chapter of differential geometry.

Under  Legendre transformation, the Cartan spaces appear as dual of the Finsler spaces \cite{Miro1}. It is remarkable that  regular Lagrangian which is 2-homogeneous in velocities is nothing but the
square of a fundamental Finsler function and its geometry is
Finsler geometry. This geometry was developed since 1918 by P.
Finsler, E. Cartan, L. Berwald, H. Akbar-Zadeh and many others, see  \cite{AZ}\cite{BCS}\cite{Matsu}\cite{ShDiff}\cite{ShLec}. Using this duality several important results in the Cartan spaces can be obtained: the canonical nonlinear
connection, the canonical metrical connection, the notion of $(\alpha, \beta)$-metrics,  etc \cite{Na}. Therefore, the
theory of Cartan spaces has the same symmetry and beauty like
Finsler geometry. Moreover, it gives a geometrical framework for
the Hamiltonian theory of Mechanics or Physical fields.

Let $(M,K)$ be a Cartan space on a manifold $M$ and put $\tau
:=\frac{1}{2}K^2$. Let us  define the symmetric $M$-tensor field
$G_{ij}:=\frac{1}{\beta}g_{ij}+\frac{v(\tau )}{\alpha \beta
}p_{i}p_{j}$ on slit cotangent bundle $T^{\ast}M_0:=T^{\ast}M-\{0\}$, where $v=v(\tau)$ is  a real valued smooth
function  defined on $[0,\infty)\subset {\mathbb R} $ and $\alpha $
and $\beta$ are real constants. Using this, we can define a Riemannian metric and almost complex structure on
$T^{\ast}M_0$ as follows
\begin{eqnarray*}
G\!\!\!\!&=&\!\!\!\!G_{ij}dx^{i}dx^{j}+G^{ij}\delta p_{i}\delta
p_{j},\\
J(\delta _{i})\!\!\!\!&=&\!\!\!\!G_{ik}\dot{\partial}^k,\ \
J(\dot{\partial}^i)=-G^{ik}\delta _k,
\end{eqnarray*}
where $G^{ij}$ is the inverse of $G_{ij}$.

In this paper, we prove that  $(T^{\ast}M_0,G,J)$ is an almost K\"{a}hlerian manifold. Then we show that the almost complex structure $J$ on $T^{\ast}M_0$ is integrable if and only if $M$ has constant scalar curvature $c$ and  the function $v$ is given  by $v=-c\alpha\beta^2$. We conclude that on a Cartan manifold $M$ of negative
constant flag curvature, $(T^{\ast}M_0,G,J)$ has  a K\"{a}hlerian structure. For Cartan manifolds of positive constant flag curvature, we show that the tube around the zero section has a K\"{a}hlerian structure (see Theorem \ref{THM1}).

Then we find the Levi-Civita connection $\nabla$ of the metric $G$. For the connection $\nabla$, we compute all of components curvature. For a Cartan space $(M,K)$ of  constant curvature $c$, we prove that
in the following cases $(M,K)$ reduce to a Riemannian space: $(i)$
for $c<0$, $(T^{\ast}M_0,G,J)$ bacame a K\"{a}hler Einstein
manifold,\ $(ii)$ for $c>0$, $(T_{\beta}^{\ast}M_0,G,J)$ became a
K\"{a}hler Einstein manifold, where $T_{\beta}^{\ast}M_0$ the tube
around the zero section in $T^{\ast}M$, defined by the condition
$2\tau<\frac{1}{c\beta^2}$. It result that, there is not any
non-Riemannian Cartan structure such that $(T^{\ast}M_0, G, J)$
became a Einstein manifold.

Finally we define divergence, gradient and Laplace  operators on the Cartan manifold $(M, K)$ with Berwald connection. Let $(M,K)$ be a Cartan space with Berwald connection,  $\textbf{S}=p^i\delta_i$ is the geodesic spray of $(M,K)$,  $X=X^i\delta_i+{\bar X}_i{\dot{\partial}}^i$ and $g:=det(g_{ij})$.  We show that $div(X)=0$ if and only if the mean Landsberg curvature of $K$ satisfies $J_i=\delta_i(\ln\sqrt{g})$. We  define the gradient operator by
$G(gradf,X)=Xf,\ \ \ \forall X\in\chi(T^{\ast}M),\ \  f\in C^{\infty}(TM)$ and prove that the gradient operator  is determined by $gradf=G^{ih}(\nabla_{\delta_h}f)\delta_i+G_{ih}(\nabla_{\dot{\partial}^h}f)\dot{\partial}^i$. The Laplace operator of a scalar field $f\in C^{\infty}(TM)$, is defined by $\Delta f=div(gradf)$.  With vanishing the Laplace operator, we prove that  $div(\textbf{S})=0$ if and only if $K$ ia a mean Landsberg metric.

\section{Preliminaries}
Let $M$ be an $n$-dimensional $C^{\infty}$ manifold and
$\pi^{\ast}:T^{\ast}M\longrightarrow M$ its cotangent bundle. If
$(x^i)$ are local coordinates on $M$, then $(x^i,p_i)$ will be taken
as local coordinates on $T^{\ast}M$ with the momenta $(p_i)$
provided by $p=p_idx^i$ where $p\in T^{\ast}_xM$, $x=(x^i)$ and
$(dx^i)$ is the natural basis of $T^{\ast}_xM$. The indices $i, j,
k,\ldots$ will run from 1 to $n$ and the Einstein convention on
summation will be used.

Put $\partial_i:=\frac{\partial}{\partial x^i}$ and $\dot{\partial}^i:=\frac{\partial}{\partial p_i}$. Let $(\partial_i, \dot{\partial}^i)$ be the natural basis in $T_{(x,p)}T^{\ast}M$ and $(dx^i,dp_i)$ be the dual
basis of it. The kernel $V_{(x,p)}$ of the differential
$d\pi^{\ast}:T_{(x,p)}T^{\ast}M\rightarrow T_xM$ is called the
\textit{vertical} subspace of $T_{(x,p)}T^{\ast}M$ and the mapping
$(x,p)\rightarrow V_{(x,p)}$ is a regular distribution on
$T^{\ast}M$ called the \textit{vertical distribution}. This is
integrable with the leaves $T^{\ast}_xM$, $x\in M$ and is locally
spanned by $\dot{\partial}^i$. The vector field
${\textbf{C}}^{\ast}=p_i\dot{\partial}^i$ is called the Liouville
vector field and $\omega=p_idx^i$ is called the Liouville 1-form on
$T^{\ast}M$. Then $d\omega$ is the canonical symplectic structure on
$T^{\ast}M$. For an easer handling of the geometrical objects on
$T^{\ast}M$,  it is usual to consider a supplementary distribution to
the vertical distribution, $(x,p)\rightarrow N_{(x,p)}$, called the
\textit{horizontal distribution} and to report all geometrical
objects on $T^{\ast}M$ to the decomposition
\begin{equation}
T_{(x,p)}T^{\ast}M=N_{(x,p)}\oplus V_{(x,p)}.\label{decom}
\end{equation}
The pieces produced by the decomposition (\ref{decom}) are called
\textit{d}-geometrical objects (\textit{d} is for distinguished)
since their local components behave like geometrical objects on $M$,
although they depend on $x=(x^i)$ and momenta $p=(p_i)$.

The horizontal distribution is taken as being locally spanned by the
local vector fields
\begin{equation}
\delta_i:=\partial_i+N_{ij}(x,p)\dot{\partial}^j.\label{decom1}
\end{equation}
The horizontal distribution is called also a nonlinear connection on $T^{\ast}M$ and the functions $(N_{ij})$ are calledthe local coefficients of this nonlinear connection. It is important
to note that any regular Hamiltonian on $T^{\ast}M$ determines a
nonlinear connection whose local coefficients verify
$N_{ij}=N_{ji}$. The basis $(\delta_i,\dot{\partial}^i)$ is adapted
to the decomposition (\ref{decom}). The dual of it is $(dx^i,\delta
P_i)$, for $\delta p_i=dp_i-N_{ji}dx^j$.

A Cartan structure on $M$ is a function  $K:T^{\ast}M\rightarrow [0,\infty)$  which has the following properties: (i) $K$ is $C^{\infty}$ on $T^{\ast}M_0=T^{\ast}M-\{0\}$; (ii) $K(x,\lambda p)=\lambda K(x,p)$
for all $\lambda>0$ and (iii) the $n\times n$ matrix $(g^{ij})$, where $g^{ij}(x,p)=\frac{1}{2}\dot{\partial}^i\dot{\partial}^jK^2(x,p)$, is
positive definite at all  points of $T^{\ast}M_0$. We notice that in
fact $K(x,p)>0$, whenever $p\neq0$. The pair $(M,K)$ is called a Cartan space.
Using this notations, let us define
\[
p^i=\frac{1}{2}\dot{\partial}^iK^2\ \  \textrm{and}\ \ C^{ijk}=-\frac{1}{4}\dot{\partial}^i\dot{\partial}^j\dot{\partial}^kK^2.
\]
The properties of $K$ imply that
\begin{eqnarray}
&&p^i=g^{ij}p_j,\ \ p_i=g_{ij}p^j,\\
&&g^{ij}p_ip_j=p_ip^j=K^2,\nonumber\\
&&C^{ijk}p_k=C^{ikj}p_k=C^{kij}p_k=0.
\end{eqnarray}
One considers the \textit{formal Christoffel symbols}
\begin{equation}
\gamma^i_{jk}(x,p):=\frac{1}{2}g^{is}(\partial_kg_{js}+\partial_jg_{sk}-\partial_sg_{jk}),
\end{equation}
and the contractions
$\gamma^{\circ}_{jk}(x,p):=\gamma^i_{jk}(x,p)p_i$,
$\gamma^{\circ}_{j\circ}:=\gamma^i_{jk}p_ip^k$. Then the functions
\begin{equation}
N_{ij}(x,p)=\gamma^{\circ}_{ij}(x,p)-\frac{1}{2}\gamma^{\circ}_{h\circ}(x,p)\dot{\partial}^hg_{ij}(x,p),\label{decom2}
\end{equation}
define a nonlinear connection on $T^{\ast}M$. This nonlinear
connection was discovered by R. Miron \cite{Mir1}. Thus a
decomposition (\ref{decom}) holds. From now on, we shall use only
the nonlinear connection given by (\ref{decom2}).

A linear connection $D$ on $T^{\ast}M$ is said to be an
$N$-\textit{linear connection} if $D$ preserves by parallelism the
distribution $N$ and $V$, also we have $D\theta=0$, for
$\theta=\delta p_i\wedge dx^i$. One proves that an $N$-linear
connection can be represented in the adapted basis
$(\delta_i,\dot{\partial}^i)$ in the form
\begin{eqnarray}
D_{\delta_j}\delta_i\!\!\!\!&=&\!\!\!\!B^k_{ij}\delta_j,\ \ \
D_{\delta_j}\dot{\partial}^i=-B^i_{kj}\dot{\partial}^k,\\
D_{\dot{\partial}^j}\delta_i\!\!\!\!&=&\!\!\!\!V^{kj}_i\delta_k,\ \
\ D_{\dot{\partial}^j}\dot{\partial}^i=-V^{ij}_k\dot{\partial}^k,
\end{eqnarray}
where $V^{kj}_i$ is a \textit{d}-tensor field and $B^k_{ij}(x,p)$
behave like the coefficients of a linear connection on $M$. The
functions $B^k_{ij}$ and $V^{kj}_i$ define operators of
\textit{h}-covariant and \textit{v}-covariant derivatives in the
algebra of \textit{d}-tensor fields, denoted by $_{|k}$ and
$\mid^k$, respectively. For $g^{ij}$, these are given by following equation
\begin{eqnarray}
{g^{ij}}_{|k}\!\!\!\!&=&\!\!\!\!\delta_kg^{ij}+g^{sj}B^i_{sk}+g^{is}B^j_{sk},\\
g^{ij}\!\!\mid^k\!\!\!\!&=&\!\!\!\!\dot{\partial}^kg^{ij}+g^{sj}V^{ik}_s+g^{is}V^{jk}_s.
\end{eqnarray}
An \textit{N}-linear connection given in the adapted basis
$(\delta_i,\dot{\partial}^j)$ as $D\Gamma(N)=(B^i_{jk},V^{ik}_j)$ is
called \textit{Berwald connection} if
\begin{equation}
{g^{ij}}_{|k}=-2L^{ij}_k,\ \ \ \ g^{ij}\!\!\mid^k=-2C^{ijk},
\end{equation}
where $L^{ij}_k={C^{ij}_k}_{|h}p^h$ are components of the Landsberg
tensor on $M$ (see  \cite{BT1}\cite{BT2}\cite{TN}\cite{TAE}).

The Berwald connection $B\Gamma(N)=(\dot{\partial}^iN_{jk},0)$ of
the Cartan spaces has the torsions d-tensors as follows
\begin{eqnarray}
&&T^i_{jk}=0,\ \  S^{jk}_i=0,\ \  V^{jk}_i=0,\ \  P^i_{jk}=0,\\
&&R_{ijk}=\delta_kN_{ij}-\delta_jN_{ik}.
\end{eqnarray}
The d-tensors of curvature of $B\Gamma(N)$ are given by
\begin{eqnarray}
R^i_{j kj}\!\!\!\!&=&\!\!\!\!\delta_hB^i_{jk}-\delta_kB^i_{jh}+B^s_{jk}B^i_{sh}-B^s_{jh}B^i_{sk}\\
P^{i   h}_{j   k}\!\!\!\!&=&\!\!\!\!\dot{\partial}^hB^i_{jk}\\
S^{ikh}_j\!\!\!\!&=&\!\!\!\!0
\end{eqnarray}
where   $B^i_{jk} = \dot{\partial}^iN_{jk}$ are the coefficients of the $B\Gamma(N)$-connection.It has also
the following properties
\begin{eqnarray}
K^2_{|j}\!\!\!\!&=&\!\!\!\!\delta_jK^2=0,\ \ K^2\!\!\mid^j=2p^j,\\
p_{i|j}\!\!\!\!&=&\!\!\!\! p^i_{|j}=0,\ \  p_i\!\!\mid^j=\delta^j_i,\ \ p^i\!\!\mid^j=g^{ij},\ \ R_{kij}p^k=0\label{4}\\
\delta_ig_{jk}\!\!\!\!&=&\!\!\!\!B^s_{ji}g_{sk}+B^s_{ki}g_{js}.\label{6}
\end{eqnarray}
\section{ K\"{a}hler Structures on Cotangent Bundle}
Suppose that
\begin{equation}
 \tau :=\frac{1}{2}K^2=\frac{1}{2}g^{ij}(x,p)p_{i}p_{j}.
\end{equation}
We consider a real valued smooth function $v$ defined on
$[0,\infty)\subset {\mathbb R} $ and real constants $\alpha $ and
$\beta$. We define the following symmetric $M$-tensor field of type
(0,2) on $T^{\ast}M_0$ having the components
\begin{equation}
G_{ij}:=\frac{1}{\beta}g_{ij}+\frac{v(\tau )}{\alpha \beta }p_{i}
p_{j}.
\end{equation}
It follows easily that the matrix  $(G_{ij})$ is positive definite
if and only if  $\alpha ,\beta >0,\, \, \,\, \alpha +2\tau v>0.$ The
inverse of this matrix has the entries
\begin{equation}
G^{kl}=\beta g^{kl}-\frac{v\beta }{\alpha +2\tau v}p^{k} p^{l}.
\end{equation}
The components  $G^{kl} $ define symmetric  $M$-tensor field of
type (0,2) on  $T^{\ast}M_0$. It is easy to see that if the matrix
$(G_{ij} )$ is positive definite, then matrix $(G^{kl} )$ is
positive definite too.

Using $(G_{ij})$ and $(G^{ij})$, the following Riemannian metric
 on $T^{\ast}M_0$ is defined
\begin{equation}
G=G_{ij}dx^{i}dx^{j}+G^{ij}\delta p_{i}\delta p_{j}.
\end{equation}

Now, we define an almost complex structure  $J$  on  $T^{\ast}M_0$
by
\begin{equation}
J(\delta _{i})=G_{ik}\dot{\partial}^k,\ \ \
J(\dot{\partial}^i)=-G^{ik}\delta _k.
\end{equation}
It is easy to check that $J^2=-I$.
\begin{theorem}
 $(T^{\ast}M_0,G,J)$ is an almost K\"{a}hlerian manifold.
\end{theorem}
\begin{proof} Since the matrix $(G^{kl} )$ is the inverse of the
matrix $(G_{ij} ),$ then we have
\[
G(J\delta_i,J\delta_j)=G_{ik}G_{jr}G(\dot{\partial}^k,\dot{\partial}^r)=G_{ik}G_{jr}G^{kr}=G_{ij}=G(\delta_i,\delta_j).
\]
The relations
\[
G(J\dot{\partial}^i,J\dot{\partial}^j)=G(\dot{\partial}^i,\dot{\partial}^j),\
\ \ \ G(J\delta _i,J\dot{\partial}^j)=G(\delta_i,\dot{\partial}^j)=0,
\]
may be obtained in a similar way, thus
\[
G(JX,JY)=G(X,Y),\, \, \, \, \, \, \, \forall X,Y\in
\Gamma(T^{\ast}M_0).
\]
It means that $G$ is almost Hermitian with respect to $J$. The
fundamental 2-form associated by this almost K\"{a}hler structure is
$\theta$, defined by
\[
\theta(X,Y):=G(X,JY),\, \, \, \, \forall X,Y\in \Gamma
(T^{\ast}M_0).
\]
Then we get
\[
\theta(\dot{\partial}^i,\delta_j)=G(\dot{\partial}^i,J\delta_j)=G(\dot{\partial}^i,G_{jk}\dot{\partial}^k)=G^{ik}G_{jk}=\delta^i_j,
\]
 and
\[
\theta(\delta_i,\delta_j)=\theta(\dot{\partial}^i,\dot{\partial}^j)=0.
\]
Hence, we have
\begin{equation}
\theta=\delta p_i\wedge dx^i,
\end{equation}
that is the canonical symplectic form of $T^{\ast}M$.
\end{proof}

Here, we study the integrability of the almost complex structure
defined by $J$ on $TM$. To do this, we need the following lemma.
\begin{lemma}{\rm (\cite{MirHri}\cite{PT1}\cite{PT2})}
\emph{Let $(M,F)$ be a Finsler manifold.  Then we have:}\\
$(1)\ \ [\delta_i,\delta_j]=R_{kij}{\dot{\partial}}^k$,\\
$(2)\ \ [\delta_i,{\dot{\partial}}^j]=-({\dot{\partial}}^jN_{ik}){\dot{\partial}}^k$,\\
$(3)\ \ [{\dot{\partial}}^i,{\dot{\partial}}^j]=0$.
\end{lemma}
\begin{lemma}\label{Im}
Let $(M,K)$ be a Cartan space. Then $J$ is complex structure on
$T^{\ast}M_0$ if and only if $A_{kij}=0$ and
\begin{equation}
R_{kij}=\frac{v}{\alpha\beta^2}(g_{ik}p_j-g_{jk}p_i),\label{2}
\end{equation}
where $A_{kij}=\delta_i G_{jk}-\delta_j
G_{ik}+G_{ir}{\dot{\partial}}^rN_{jk}-G_{jr}{\dot{\partial}}^rN_{ik}$.
\end{lemma}
\begin{proof}
Using the definition of the Nijenhuis tensor field $N_J$ of $J$, that is,
\[
N_J(X,Y)=[JX,JY]-J[JX,Y]-J[X,JY]-[X,Y],\ \ \ \forall X, Y\in
\Gamma(T^{\ast}M)
\]
we get
\begin{equation}
N_J(\delta_i,\delta_j)=A_{hij}G^{hk}\delta_k+
(M_{kij}-R_{kij}){\dot{\partial}}^k,\label{Nij1}
\end{equation}
where
$M_{kij}=G_{ir}{\dot{\partial}}^rG_{jk}-G_{jr}{\dot{\partial}}^rG_{ik}$. Let $C_{jk}^r:=g_{jl}g_{sk}C^{rls}$, then we have
\[
{\dot{\partial}}^rg_{jk}=-g_{jl}g_{sk}{\dot{\partial}}^rg^{ls}=2g_{jl}g_{sk}C^{rls}=2C_{jk}^r.
\]
By above equation, we obtain
\begin{equation}
G_{ir}{\dot{\partial}}^rG_{jk}=\frac{2}{\beta^2}C_{ijk}
+\frac{v}{\alpha\beta^2}(g_{ji}p_k+g_{ik}p_j)
+(\frac{v^{\prime}}{\alpha\beta^2}+\frac{2vv^{\prime}\tau}{\alpha\beta}+\frac{2v^2}{\alpha^2\beta^2})p_ip_jp_k.\label{Nij2}
\end{equation}
where $C_{ijk}=g_{ir}C_{jk}^r$. From (\ref{Nij2}) we get
\begin{equation}
M_{kij}=\frac{v}{\alpha\beta^2}(g_{ik}p_j-g_{jk}p_i).\label{Nij3}
\end{equation}
By a straightforward computation, it follows that
$N_J({\dot{\partial}}^i,{\dot{\partial}}^j)=0,
N_J({\dot{\partial}}^i,\delta_j)=0$, whenever
$N_j(\delta_i,\delta_j)=0$. Therefore, from relations (\ref{Nij1})
and (\ref{Nij3}), we conclude that the necessary and sufficient
conditions for the Nijenhuis tensor field $N_J$ to vanish, so that J
is a complex structure, are that $A_{kij}=0$ and (\ref{2}) hold.
\end{proof}

In equation (\ref{2}), we put $-\frac{v}{\alpha\beta^2}=c$, where $c$
is constant. Then we get
\begin{equation}
R_{kij}=c(g_{jk}p_i-g_{ik}p_j).\label{3}
\end{equation}
\begin{theorem}\label{Im1}
Let $(M,K)$ be a Cartan space of dimension $n\geq 3$. Then the
almost complex structure $J$ on $T^{\ast}M_0$ is integrable if and
only if (\ref{3}) is hold and the function $v$ is given by
\begin{equation}
v=-c\alpha\beta^2.\label{constant}
\end{equation}
\end{theorem}
\begin{proof}
From equation $p_{i|k}=0$ of relation (\ref{4}), we conclude that
$\delta_ip_k=N_{ik}$. Hence we obtain
\begin{eqnarray}
A_{kij}\!\!\!\!&=&\!\!\!\!\delta_i g_{jk}-\delta_j
g_{ik}+g_{ir}{\dot{\partial}}^rN_{jk}-g_{jr}{\dot{\partial}}^rN_{ik}\nonumber\\
\!\!\!\!&=&\!\!\!\!\delta_ig_{jk}-\delta_jg_{ik}+g_{ir}B^r_{jk}-g_{jr}B^r_{ik}\nonumber\\
\!\!\!\!&=&\!\!\!\!g_{jk|i}-g_{ik|j}\nonumber\\
\!\!\!\!&=&\!\!\!\!2L_{jki}-2L_{ikj}=0
\end{eqnarray}
Now we suppose that  $v=-c\alpha\beta^2$. Thus from equation $A_{kij}=0$ and
Lemma \ref{Im}, we conclude that $J$ is integrable if and only if
(\ref{3}) is hold.
\end{proof}
A Cartan space $K^n$ is of constant scalar curvature $c$ if
\begin{equation}
H_{hijk}p^ip^jX^hX^k=c(g_{hj}g_{ik}-g_{hk}g_{ij})p^ip^jX^hX^k,\label{7}
\end{equation}
for every $(x,p)\in T^{\ast}_0M$ and $X=(X^i)\in T_xM$. Here
$H_{hijk}$ is the (hh)h-curvature of the linear Cartan connection of
$K^n$. We replace $H_{hijk}$ in (\ref{7}) with $g_{is}H^s_{hjk}$ and so it
reduce to following
\begin{equation}
p_sH^s_{hjk}p^jX^hX^k=c(p_hp_k-K^2g_{hk})X^hX^k.\label{8}
\end{equation}
By part (ii) of Proposition 5.1 in chapter 7 of \cite{MirHri},
$p_sH^s_{hjk}=-R_{hjk}$, hence we get
\[
R_{hjk}p^jX^hX^k=c(K^2g_{hk}-p_hp_k)X^hX^k,
\]
or equivalently
\begin{equation}
R_{hjk}p^j=c(K^2g_{hk}-p_hp_k),\label{8}
\end{equation}
because $(X^h)$ and $X^k$ are arbitrary vector fields on $M$. It is
easy to check that (\ref{8}) follows from (\ref{3}). Similarly can
be shown that if Cartan space $K^n$ has the constant scalar
curvature $c$, then the equation (\ref{3}) is hold (see
\cite{Matsu}).
\begin{theorem}\label{THM1}
Let $(M,K)$ be a Cartan space with constant flag curvature $c$. Suppose that $v$ is given by (\ref{constant}). Then\\
$(i)$ for negative constant $c$, structure $(T^{\ast}M_0,G,J)$ is a K\"{a}hler manifold;\\
$(ii)$ for positive constant $c$, the tube around the zero section
in $T^{\ast}M$, defined by the condition $2\tau=K^2<\frac{1}{c\beta^2}$,  is a K\"{a}hler manifold.
\end{theorem}
\begin{proof}
The function  $v$  must satisfies in the following condition
\begin{equation}
\alpha +2\tau v=\alpha (1+2(-c)\beta ^{2} \tau )>0,\, \, \, \, \,
\alpha ,\beta >0.
\end{equation}
By using the above relation and Theorem \ref{Im1}, we complete the proof.
\end{proof}
By attention to the Theorem \ref{THM1}, the components of the K\"{a}hler
metric $G$ on $T^{\ast}M_0$ are
\begin{equation}
\left\{\begin{array}{l} {G_{ij} =\frac{1}{\beta } g_{ij} -c\beta
p_{i}p_{j} ,} \\ {G^{ij} =\beta g_{ij}+\frac{c\beta ^{3} }{1-2c\beta
^{2} \tau } p_{i}p_{j} .}
\end{array}\right.\label{Kahler}
\end{equation}
\section{A K\"{a}hler Einstein Structure on Cotangent Bundle}
In this section, we study the property of $(T^{\ast}M_0,G)$ to be
Einstein. We find the expression of the Levi-Civita connection $\nabla$ of the metric $G$ on $T^{\ast}M_0$, then we get  the curvature tensor field of $\nabla$. Then, by computing the corresponding traces, we find the components of Ricci tensor field of   $\nabla$.
\subsection{The Levi-Civita Connection}
\begin{lemma}
The Levi-Civita connection of the K\"{a}hler metric $G$ are given by following
\begin{eqnarray}
\nabla_{{\dot{\partial}}^i}{\dot{\partial}}^j\!\!\!\!&=&\!\!\!\!(\beta^2L^{ijs})\delta_s+(-C^{ij}_s+c\beta
G^{ij}p_s){\dot{\partial}}^s, \label{con3}\\
\nabla_{\delta_i}{\dot{\partial}}^j\!\!\!\!&=&\!\!\!\!(C^{js}_i-c\beta
G^{js}p_i)\delta_s-(L^j_{is}+B^j_{is}){\dot{\partial}}^s,\label{con2}\\
\nabla_{{\dot{\partial}}^i}{\delta_j}\!\!\!\!&=&\!\!\!\!(C^{is}_j-c\beta
G^{is}p_j)\delta_s-L^i_{js}{\dot{\partial}}^s,\label{con4}\\
\nabla_{\delta_i}\delta_j\!\!\!\!&=&\!\!\!\!(L^s_{ij}+B^s_{ij})\delta_s+(-\frac{1}{\beta^2}C_{ijs}+c\beta
G_{js}p_i){\dot{\partial}}^s.\label{con9}
\end{eqnarray}
\end{lemma}
\begin{proof}
Recall that for Cartan space with Berwald connection, the relation
$B^j_{ik}={\dot{\partial}}^jN_{ik}$ is hold, and so we have
$[\delta_i,{\dot{\partial}}^j]=B^j_{ik}{\dot{\partial}}^k$. Also the
Levi-Civita connection $\nabla$ of the Riemannian manifold
$(T^{\ast}M_0,G)$ is obtained from the formula
\begin{eqnarray}
2G({\nabla }_{X} Y,Z)\!\!\!\!&=&\!\!\!\!X(G(Y,Z))+Y(G(X,Z))-Z(G(X,Y))\nonumber\\
\!\!\!\!&+&\!\!\!\!G([X,Y],Z)-G([X,Z],Y)-G([Y,Z],X),\, \, \, \, \,
\forall X,Y,Z\in \Gamma(T^{\ast}M_0),\label{Levi}
\end{eqnarray}
and is characterized by the conditions $\nabla G=0$ and $T=0$, where
$T$ is the torsion tensor of $\nabla$. Let
$\nabla_{{\dot{\partial}}^i}{\dot{\partial}}^j=\Gamma^{ijh}\delta_h+\Gamma^{ij}_h{\dot{\partial}}^h$. Then we get
\begin{eqnarray}
\Gamma^{ijs}\!\!\!\!&=&\!\!\!\!\frac{1}{2}[-\delta_k(\beta
g^{ij}+\frac{c\beta^3}{1-2c\beta^2\tau}p^ip^j)-B^i_{km}(\beta
g^{mj}+\frac{c\beta^3}{1-2c\beta^2\tau}p^mp^j)\nonumber\\
\!\!\!\!&&\!\!\!\!-B^j_{km}(\beta
g^{mi}+\frac{c\beta^3}{1-2c\beta^2\tau}p^mp^i)](\beta
g^{ks}+\frac{c\beta^3}{1-2c\beta^2\tau}p^kp^s)=\beta^2g^{ks}L^{ij}_k.
\end{eqnarray}
Similarly we obtain
\begin{eqnarray}
\Gamma^{ij}_s\!\!\!\!&=&\!\!\!\!\frac{1}{2}[{\dot{\partial}}^i(\beta
g^{jk}+\frac{c\beta^3}{1-2c\beta^2\tau}p^jp^k)+{\dot{\partial}}^j(\beta
g^{ik}+\frac{c\beta^3}{1-2c\beta^2\tau}p^ip^k)\nonumber\\
\!\!\!\!&&\!\!\!\!-{\dot{\partial}}^k(\beta
g^{ij}+\frac{c\beta^3}{1-2c\beta^2\tau}p^ip^j)](\frac{1}{\beta}g_{ks}-c\beta
p_kp_s)=-C^{ij}_s+c\beta G^{ij}p_s.
\end{eqnarray}
Using two above equation, we have (\ref{con3}). By a similar way,
we obtain (\ref{con2}), (\ref{con4}) and (\ref{con9}).
\end{proof}
We say that the vertical distribution $VT^{\ast}M_0$ is totally
geodesic (resp. minimal) in $TT^{\ast}M_0$ if
$H\nabla_{\dot{\partial}^i}\dot{\partial}^j=0$ (resp.
$g_{ij}H\nabla_{\dot{\partial}^i}\dot{\partial}^j=0$), where $H$
denotes the horizontal projection. Similarly, if we denote by $V$
the vertical projection, then we say that the horizontal
distribution $HT^{\ast}M_0$ is totally geodesic (resp. minimal) in
$TT^{\ast}M_0$ if $V\nabla_{\delta_i}\delta_j=0$ (resp.
$g^{ij}V\nabla_{\delta_i}\delta_j=0$). By using (\ref{con3}), we
obtain
\begin{equation}
H\nabla_{{\dot{\partial}}^i}{\dot{\partial}}^j=\beta^2L^{ijs}\delta_s,
\end{equation}
and
\begin{equation}
g_{ij}H\nabla_{{\dot{\partial}}^i}{\dot{\partial}}^j=\beta^2g_{ij}L^{ijs}\delta_s=\beta^2J^s\delta_s,
\end{equation}
where $J^s$ is the mean Landsberg tensor. Hence, we have the
following.
\begin{cor}
Let $(M,K)$ be a Cartan space with Berwald connection. Then we have\\
$(i)$ $K$ is Landsberg metric if and only if the vertical
distribution $VT^{\ast}M_0$ is totally geodesic in $TT^{\ast}M_0$;\\
$(ii)$ $K$ is weakly Landsberg metric if and only if the vertical
distribution $VT^{\ast}M_0$ is minimal in $TT^{\ast}M_0$.
\end{cor}

\begin{cor}
The horizontal distribution $HT^{\ast}M_0$ can not be totally
geodesic or minimal in $TT^{\ast}M_0$.
\end{cor}
\begin{proof}
By (\ref{con9}), we have
\[
V\nabla_{\delta_i}\delta_j=(-\frac{1}{\beta^2}C_{ijs}+c\beta
G_{js}p_i){\dot{\partial}}^s.
\]
If $HT^*M$ is totally geodesic, then we have $c\beta
G_{js}p_ip^j=0$. Therefore, we obtain
\[
cp_ip_s(1-2c\beta^2\tau)=0,
\]
which can not be true.
\end{proof}
\subsection{The Curvature Tensors}
\begin{theorem}
The coefficients of the curvature tensor of K\"{a}hler metric $G$ as
follows
\begin{eqnarray}
K({\dot{\partial}}^i,{\dot{\partial}}^j){\dot{\partial}}^k\!\!\!\!&=&\!\!\!\!
\Big[\beta^2(L^{jkh}{|^i}-L^{ikh}{|^j})\Big]\delta_h
+\Big[C^{ik,j}_h-C^{jk,i}_h+c\beta G^{jk}\delta^i_h-c\beta G^{ik}\delta^j_h+C^{jk}_sC^{is}_h\nonumber\\
\!\!\!\!&&\!\!\!\!-C^{ik}_sC^{js}_h+\beta^2(L^j_{sh}L^{sik}-L^i_{sh}L^{sjk})\Big]{\dot{\partial}}^h,\label{Curvature6}
\end{eqnarray}
\begin{eqnarray}
\nonumber K(\delta_i,{\dot{\partial}}^j){\dot{\partial}}^k\!\!\!\!&=&\!\!\!\!\Big[c\beta
G^{kh}\delta^j_i-C^{kh,j}_i-C^{jh}_sC^{ks}_i-C^{jk}_sC^{hs}_i+\beta^2( L^{sjk}L^h_{is}+{L^{hjk}}_{|i}+L^k_{si}L^{hjs})\Big]\delta_h\\
\nonumber \!\!\!\!&&\!\!\!\!+\Big[B^{k,j}_{ih}-{C^{jk}_h}_{|i}-c\beta^2L^{jk}_ip_h-C_{ish}L^{jsk}
+C^{jk}_sL^s_{ih}+C^{sk}_iL^j_{sh}\\
\!\!\!\!&&\!\!\!\!-C^{js}_hL^k_{is}+L^{k,j}_{hi}
\Big]{\dot{\partial}}^h,\label{curvature13}
\end{eqnarray}
\begin{eqnarray}
\nonumber K(\delta_i,\delta_j){\delta_k}
\!\!\!\!&=&\!\!\!\!\Big[R^h_{kji}+\frac{1}{\beta^2}(C_{iks}C^{hs}_j-C_{jks}C^{hs}_i)+
c^2\beta^2(p_i\delta^h_j-p_j\delta^h_i)p_k+(L^s_{kj}L^h_{is}
-L^s_{ki}L^h_{js})\\ \nonumber \!\!\!\!&&\!\!\!\!+(L^h_{kj|i}-L^h_{ki|j})\Big]\delta_h+\Big[\frac{1}{\beta^2}(C_{ikh|j}-C_{jkh|i})+2R_{sij}L^s_{hk}+
\frac{1}{\beta^2}(C_{jks}L^s_{ih}-C_{iks}L^s_{jh}
\nonumber\\\!\!\!\!&&\!\!\!\!
+C_{jhs}L^s_{ki}-C_{ihs}L^s_{jk})\Big]{\dot{\partial}}^h,\label{curvature14}
\end{eqnarray}
\begin{eqnarray}
K(\delta_i,\delta_j){\dot{\partial}}^k\!\!\!\!&=&\!\!\!\!\Big[{C^{kh}_j}_{|i}-{C^{kh}_i}_{|j}+c\beta^2(p_jL^{kh}_i-p_iL^{kh}_j)+C^{ks}_jL^h_{si}-
C^{ks}_iL^h_{sj}\nonumber\\
\!\!\!\!&&\!\!\!\!+C^{sh}_jL^k_{si}-C^{sh}_iL^k_{sj}\Big]\delta_h+\Big[-R^k_{hji}+\frac{1}{\beta^2}(C^{ks}_iC_{jhs}-C^{ks}_jC_{ihs})+\nonumber\\
\!\!\!\!&&\!\!\!\!c^2\beta^2p_h(p_j\delta^k_i-p_i\delta^k_j)+L^{k}_{hi|j}-L^{k}_{hj|i}
+L^k_{sj}L^s_{hi}-L^k_{si}L^s_{hj}\Big]{\dot{\partial}}^h,\label{curvature15}
\end{eqnarray}
\begin{eqnarray}
K({\dot{\partial}}^i,{\dot{\partial}}^j)\delta_k\!\!\!\!&=&\!\!\!\!\Big[C^{jh,i}_k-C^{ih,j}_k+C^{js}_kC^{ih}_s-C^{is}_kC^{jh}_s+c\beta(
G^{ih}\delta^j_k-G^{jh}\delta^i_k)\nonumber\\
\!\!\!\!&&\!\!\!\!+\beta^2(L^{jsh}L^i_{sk}-L^{ish}L^j_{sk})\Big]\delta_h
+\Big[L^i_{kh}|^j-L^j_{kh}|^i\Big]{\dot{\partial}}^h,\label{curvature16}
\end{eqnarray}
\begin{eqnarray}
K(\delta_i,{\dot{\partial}}^j)\delta_k\!\!\!\!&=&\!\!\!\!\Big[C^{jh}_{k|i}+c\beta^2L^{jh}_ip_k
-L^{h,j}_{ki}-B^{h,j}_{ik}+C^{js}_kL^h_{si}\nonumber\\
\!\!\!\!&&\!\!\!\!-C^{jh}_sL^s_{ki}-C^{sh}_iL^j_{sk
}+C_{iks}L^{hjs}\Big]\delta_h\nonumber\\
\!\!\!\!&&\!\!\!\!+\Big[\frac{1}{\beta^2}(C^{,j}_{ikh}-C_{ish}C^{js}_k-C_{iks}C^{js}_h)+cp_hC^j_{ik}+cp_kC^j_{ih}
\nonumber\\
\!\!\!\!&&\!\!\!\!-c\beta
G_{kh}\delta^j_i+L^j_{sk}L^s_{hi}+L^j_{sh}L^s_{ki}-L^j_{hk|i}\Big]{\dot{\partial}}^h.\label{curvature17}
\end{eqnarray}
\end{theorem}
\begin{proof}
Recall that the curvature $K$ of $\nabla$ is obtained from the following
\begin{equation}
K(X,Y)Z={\nabla }_{X} {\nabla }_{Y} Z-{\nabla }_{Y}{\nabla }_{X}
Z-{\nabla }_{[X,Y]} Z,\ \ \  \forall
X,Y,Z\in\Gamma(TM).\label{Curvature1}
\end{equation}
Using (\ref{Curvature1}) we have
\begin{equation}
K({\dot{\partial}}^i,{\dot{\partial}}^j){\dot{\partial}}^k={\nabla
}_{{\dot{\partial}}^i} {\nabla
}_{{\dot{\partial}}^j}{\dot{\partial}}^k-{\nabla
}_{{\dot{\partial}}^j}{\nabla
}_{{\dot{\partial}}^i}{\dot{\partial}}^k .\label{Curvature2}
\end{equation}
By (\ref{con3}), it follows that
\begin{eqnarray}
\nabla_{{\dot{\partial}}^i}\nabla_{{\dot{\partial}}^j}{\dot{\partial}}^k\!\!\!\!&=&\!\!\!\!
{\dot{\partial}}^i(\beta^2g^{ms}L^{jk}_m)\delta_s
+(\beta^2g^{ms}L^{jk}_m)\nabla_{{\dot{\partial}}^i}\delta_s\nonumber\\
\!\!\!\!&+&\!\!\!\!{\dot{\partial}}^i(-C^{jk}_s+c\beta
G^{jk}p_s){\dot{\partial}}^s+(-C^{jk}_s+c\beta
G^{jk}p_s)\nabla_{{\dot{\partial}}^i}{\dot{\partial}}^s.
\label{Curvature3}
\end{eqnarray}
Since $p_hL^h_{rl}=0$, $p^lL^h_{rl}=p^rL^h_{rl}=0$ and
${\dot{\partial}}^kg^{ij}=-2C^{kij}$,  then by (\ref{Curvature3}) the following relation yields
\begin{eqnarray}
\nabla_{{\dot{\partial}}^i}\nabla_{{\dot{\partial}}^j}{\dot{\partial}}^k\!\!\!\!&=&\!\!\!\!
[\beta^2L^{hjk,i}+\beta^2(C^{mih}L^{jk}_m-C^{jk}_sL^{ihs})]\delta_h\nonumber\\
\!\!\!\!&&\!\!\!\!
+[-\beta^2g^{ms}g_{nh}L^{jk}_mL^{in}_s+C^{jk}_sC^{is}_h+c^2\beta^2G^{jk}G^{is}p_sp_h\nonumber\\
\!\!\!\!&&\!\!\!\!
-c\beta^2C^{ijk}p_h+{\dot{\partial}}^i(-C^{jk}_h+c\beta
G^{jk}p_h)]{\dot{\partial}}^h,\label{Curvature4}
\end{eqnarray}
where $L^{hjk,i}={\dot{\partial}}^iL^{hjk}$. Since
${\dot{\partial}}^i\tau=p^i$ and ${\dot{\partial}}^ip^j=g^{ij}$,
then we obtain
\begin{eqnarray}
c\beta G^{jk,i}p_s-\!\!\!\!&&\!\!\!\!c\beta G^{ik,j}p_s+c^2\beta^2G^{jk}G^{ih}p_hp_s-c^2\beta^2G^{ik}G^{jh}p_hp_s=\nonumber\\
\!\!\!\!&&\!\!\!\!\frac{c^2\beta^4}{1-2c\beta^2\tau}(g^{ik}p^j-g^{jk}p^i+g^{jk}p^i-g^{ik}p^j)=0.\hspace{2.8cm}\label{Curvature5}
\end{eqnarray}
With replace $i,j$ in (\ref{Curvature4}) and setting this equations
in (\ref{Curvature2}), also by attention (\ref{Curvature5}), we get
\begin{eqnarray}
K({\dot{\partial}}^i,{\dot{\partial}}^j){\dot{\partial}}^k\!\!\!\!&=&\!\!\!\!
[\beta^2(L^{jkh}{|^i}-L^{ikh}{|^j})]\delta_h
+[C^{ik,j}_h-C^{jk,i}_h+c\beta G^{jk}\delta^i_h-c\beta G^{ik}\delta^j_h+C^{jk}_sC^{is}_h\nonumber\\
\!\!\!\!&&\!\!\!\!-C^{ik}_sC^{js}_h+\beta^2(L^j_{sh}L^{sik}-L^i_{sh}L^{sjk})]{\dot{\partial}}^h.\label{Curvature6}
\end{eqnarray}
Similarly we can obtain the other components of curvature tensor.
\end{proof}
\begin{theorem}
Let $(M,K)$ be a Cartan space of constant curvature $c$ and the
components of the metric $G$ are given by (\ref{Kahler}).
Then the following are hold if and only if $(M,K)$ is reduce to a Riemannian space.\\
$(i)$ for $c<0$, $(T^{\ast}M_0,G,J)$ is a K\"{a}hler Einstein manifold.\\
$(ii)$ for $c>0$, $(T_{\beta}^{\ast}M_0,G,J)$ is a K\"{a}hler
Einstein manifold, where $T_{\beta}^{\ast}M_0$ the tube around the
zero section in $TM$, defined by the condition
$2\tau<\frac{1}{c\beta^2}$.
\end{theorem}
\begin{proof}
Let $(M,K)$ is  a Riemannian space. Then $C^{hi}_k$ and $P^h_{ik}$
are vanish and $H^i_{jk}$ is a function of $(x^h)$. Therefore
(\ref{curvature14}) reduces to following
\begin{equation}
K(\delta_i,\delta_j)\delta_k=[R^s_{kji}+c^2\beta^2(p_i\delta^s_j-p_j\delta^s_i)p_k]\delta_s.\label{Ricc1}
\end{equation}
From Proposition 10.2 in chapter 4 of \cite{MirHri}, we have
$R_{kji}=-p_hR^h_{kji}$. Then we have
\begin{equation}
p_hR^h_{kji}=c(g_{kj}\delta^h_i-g_{ki}\delta^h_j)p_h.\label{Ricc2}
\end{equation}
Differentiating (\ref{Ricc2}) with respect to $p_s$ and taking $p=0$,  follows that
\begin{equation}\label{alaki}
R^s_{kji}=c(g_{kj}\delta^s_i-g_{ki}\delta^s_j).
\end{equation}
By putting (\ref{alaki}) in (\ref{Ricc1}),  one can obtains
\begin{eqnarray}
K(\delta_i,\delta_j)\delta_k\!\!\!\!&=&\!\!\!\!c\beta[(\frac{1}{\beta}g_{kj}-c\beta
p_kp_j)\delta^s_i-(\frac{1}{\beta}g_{ki}-c\beta\
p_kp_i)\delta^s_j]\delta_s,\nonumber\\
\!\!\!\!&=&\!\!\!\!c\beta(G_{kj}\delta^s_i-G_{ki}\delta^s_j)\delta_s.\label{Ricc3}
\end{eqnarray}
Also from (\ref{curvature17}), we get
\begin{equation}
K({\dot{\partial}}^i,\delta_j)\delta_k=c\beta
G_{sk}\delta^i_j{\dot{\partial}}^s.\label{Ricc4}
\end{equation}
From (\ref{Ricc3}) and (\ref{Ricc4}), we conclude that
\begin{eqnarray}
Ric(\delta_j,\delta_k)\!\!\!\!&=&\!\!\!\!G^{hi}G(K(\delta_i,\delta_j)\delta_k,\delta_h)+G_{hi}G(K({\dot{\partial}}^i,\delta_j)\delta_k,{\dot{\partial}}^h),\nonumber\\
\!\!\!\!&=&\!\!\!\!c\beta(G_{kj}\delta^s_i-G_{ki}\delta^s_j)G^{hi}G_{sh}+c\beta
G_{sk}\delta^i_jG_{hi}G^{sh}\nonumber\\
\!\!\!\!&=&\!\!\!\!cn\beta G_{jk}\nonumber\\
\!\!\!\!&=&\!\!\!\!cn\beta
G(\delta_j,\delta_k).\label{Ricc5}
\end{eqnarray}
Similarly from (\ref{Curvature6}) and (\ref{curvature13}),
respectively, it follows that
\begin{eqnarray}
K({\dot{\partial}}^i,{\dot{\partial}}^j){\dot{\partial}}^k=c\beta(G^{jk}\delta^i_s-G^{ik}\delta^j_s){\dot{\partial}}^s,\label{Ricc6}
\end{eqnarray}
and
\begin{eqnarray}
K(\delta_i,{\dot{\partial}}^j){\dot{\partial}}^k=c\beta
G^{ks}\delta^j_i\delta_s.\label{Ricc7}
\end{eqnarray}
By using (\ref{Ricc6}) and (\ref{Ricc7}), we obtain
\begin{eqnarray}
Ric({\dot{\partial}}^j,{\dot{\partial}}^k)\!\!\!\!&=&\!\!\!\!G^{ih}G(K(\delta_i,{\dot{\partial}}^j){\dot{\partial}}^k,\delta_h)
+G_{ih}G(K({\dot{\partial}}^i,{\dot{\partial}}^j){\dot{\partial}}^k,{\dot{\partial}}^h)\nonumber\\
\!\!\!\!&=&\!\!\!\!c\beta G^{ks}\delta^j_iG^{ih}G_{hs}+c\beta(G^{jk}\delta^i_s-G^{ik}\delta^j_s)G_{hi}G^{hs}\nonumber\\
\!\!\!\!&=&\!\!\!\!cn\beta G^{jk}\nonumber\\
\!\!\!\!&=&\!\!\!\ cn\beta G({\dot{\partial}}^j,{\dot{\partial}}^k).\label{Ricc8}
\end{eqnarray}
From (\ref{curvature13}) and (\ref{curvature15}), we have,
respectively
\begin{equation}
K(\delta_i,\delta_j){\dot{\partial}}^k=(R^k_{sij}+c\beta
R_{hij}G^{hk}p_s){\dot{\partial}}^s,\label{Ricc9}
\end{equation}
and
\begin{equation}
K({\dot{\partial}}^i,\delta_j){\dot{\partial}}^k=-c\beta
G^{ks}\delta^i_j\delta_s.\label{Ricc10}
\end{equation}
By using (\ref{Ricc9}) and (\ref{Ricc10}), we obtain
\begin{equation}
Ric(\delta_j,{\dot{\partial}}^k)=G^{ih}G(K(\delta_i,\delta_j){\dot{\partial}}^k,\delta_h)
+G_{ih}G(K({\dot{\partial}}^i,\delta_j){\dot{\partial}}^k,{\dot{\partial}}^h)=0.\label{Ricc11}
\end{equation}
From (\ref{curvature16}), we get
\begin{equation}
K({\dot{\partial}}^i,{\dot{\partial}}^j)\delta_k=c\beta(G^{is}\delta^j_k-G^{js}\delta^i_k)\delta_s.\label{Ricc12}
\end{equation}
By attention to (\ref{Ricc4}) and (\ref{Ricc12}), one can yields
\begin{equation}
Ric({\dot{\partial}}^j,\delta_k)=G^{ih}G(K(\delta_i,{\dot{\partial}}^j)\delta_k,\delta_h)
+G_{ih}G(K({\dot{\partial}}^i,{\dot{\partial}}^j)\delta_k,{\dot{\partial}}^h)=0.\label{Ricc13}
\end{equation}
From (\ref{Ricc5}), (\ref{Ricc8}), (\ref{Ricc11}) and
(\ref{Ricc13}), it follows that $Ric(X,Y)=cn\beta G(X,Y)$, $\forall X, Y\in\chi(T^{\ast}M)$.  This means that $(T^{\ast}M,G)$ is a Einstein manifold. Conversely, let $(i), (ii)$ are hold. Then there exist constant $\lambda$ such that $Ric(X,Y)=\lambda G(X,Y)$. We consider following cases:\\\\
\textbf{Case (1)}. If $\lambda=0$ (i.e., $(T^{\ast}M,G)$ is Ricci
flat), then we have $Ric({\dot{\partial}}^j,{\dot{\partial}}^k)=0$.
By using (\ref{Curvature6}) and (\ref{curvature13}) we get
\[
p_kG_{ih}G(K({\dot{\partial}}^i,{\dot{\partial}}^j){\dot{\partial}}^k,{\dot{\partial}}^h)=p_kC^{hk,j}_h-p_kC^{jk,h}_h+(n-1)c\beta
p_kG^{jk}
\]
and
\[
G^{ih}G(K(\delta_i,{\dot{\partial}}^j){\dot{\partial}}^k,\delta_h)=c\beta
p_kG^{jk}-p_kC^{kh,j}_h+\beta^2p_k{L^{hjk}}_{|h}
\]
By using two above equation, it results that
\begin{equation}
0=p_kRic({\dot{\partial}}^j,{\dot{\partial}}^k)=cn\beta
p_kG^{jk}-p_kC^{jk,h}_h+\beta^2p_k{L^{hjk}}_{|h}.\label{Ricc14}
\end{equation}
With a simple calculation, one can obtains
\begin{equation}
cn\beta p_kG^{jk}=cn\beta p_k(\beta
g^{jk}+\frac{c\beta^3}{1-2c\beta^2\tau}p^jp^k)=\frac{cn\beta^2}{1-2c\beta^2\tau}p^j,\label{Ricc17}
\end{equation}
and
\begin{eqnarray}
\!\!\!\!&&\!\!\!\!p_kC^{jk,h}_h=-p^{,h}_kC^{jk}_h=-\delta^h_kC^{jk}_h=-C^{jh}_h=-I^j,\\ \!\!\!\!&&\!\!\!\!p_k{L^{hjk}}_{|h}=-p_{k|h}L^{hjk}=0.\label{Ricc18}
\end{eqnarray}
By using (\ref{Ricc14})-(\ref{Ricc18}), we obtain
\begin{equation}
\frac{cn\beta^2}{1-2c\beta^2\tau}p^j+I^j=0.\label{Ricc19}
\end{equation}
Since $p_jI^j=0$, then by contracting (\ref{Ricc19}) with $p_j$,
we have $\frac{2cn\beta^2\tau}{1-2c\beta^2\tau}=0$. Thus we get $\beta=0$,  which  is a contradiction.\\\\
\textbf{Case (2)}. If $\lambda\neq0$, then we have
$p_kRic({\dot{\partial}}^j,{\dot{\partial}}^k)=\lambda G^{jk}p_k$.
By using  (\ref{curvature13}), (\ref{Curvature6}), (\ref{Ricc17}) and (\ref{Ricc18}), we obtain
\begin{equation}
I^j=(\lambda-cn\beta)\frac{\beta}{1-2c\beta^2\tau}p^j.\label{Ricc21}
\end{equation}
Contracting (\ref{Ricc21}) with $p_j$ yields
\begin{equation}
(\lambda-cn\beta)\frac{2\beta\tau}{1-2c\beta^2\tau}=0,\label{Ricc22}
\end{equation}
i.e., $\lambda=cn\beta$. Thus by (\ref{Ricc21}), we conclude  that
$I^j=0$, i.e., $(M,K)$ is reduces to a Riemannian space.
\end{proof}
\begin{cor}
There is not any non-Riemannian Cartan structure such that
$(T^{\ast}M_0, G, J)$ became a Einstein manifold.
\end{cor}
\section{Divergence, Gradient and Laplace  Operators}
The divergence and Laplace operator have a number of applications for study various electromagnetic, gravitational and diffusion processes.
 For instance, in general relativity theory they are uniquely defined by the Levi-Civita connection for a corresponding fixing of local frames of coordinates.
 Such constructions are naturally generalized on Finsler spaces if we work with the Cartan distinguished connection because it is also metric compatible and completely defined
 by the metric structure. Even such a linear connection contains nontrivial torsion components (uniquely determined by some prescribed metric and nonlinear connection structures),
 the torsion contribution can be encoded into some divergence terms,
 for instance, in the case of stochastic/diffusion processes. How to define in a unique self-consistent form the divergence and Laplace operators,
   in Finsler-Lagrange and Hamilton-Cartan geometries with nonmetricity, without involving the Cartan distinguished connection
   (for instance, for the Chern and/or Berwald distinguished connections) it is an unsolved task in modern mathematical physics (see \cite{V}, \cite{V0}).
\begin{proposition}\label{div}
Let $(M,K)$ be a Cartan space with Berwald connection. Then we have,
\begin{equation}
div(X^V)=0,\ \ \ \ div(X^H)=X^i\delta_i(\ln\sqrt{g})-X^iJ_i,
\end{equation}
where $X=X^i\delta_i+{\bar X}_i{\dot{\partial}}^i$ and $g:=det(g_{ij})$.
\end{proposition}
\begin{proof}
By a simple calculation, we have
\begin{eqnarray*}
div({\dot{\partial}}^i)\!\!\!\!&=&\!\!\!\!G^{jl}G(\nabla_{\delta_i}{\dot{\partial}}^j,\delta_l)
+G_{jl}G(\nabla_{{\dot{\partial}}^j}{\dot{\partial}}^i,{\dot{\partial}}^l)\\
\!\!\!\!&=&\!\!\!\! G^{jl}(C^{is}_j-c\beta G^{is}p_j)G_{sl}
+G_{jl}(-C^{ji}_s+c\beta
G^{ji}p_s)G^{sl}\\
\!\!\!\!&=&\!\!\!\!C^{is}_s-c\beta G^{is}p_s+c\beta
G^{is}p_s-C^{is}_s=0,\\
div(\delta_i)\!\!\!\!&=&\!\!\!\!G^{jl}G(\nabla_{\delta_i}\delta_j,\delta_l)
+G_{jl}G(\nabla_{\dot{\partial}}^j{\delta_i},{\dot{\partial}}^l)\\
\!\!\!\!&=&\!\!\!\!G^{jl}(L^s_{ij}+B^s_{ij})G_{sl}+G_{jl}(-L^i_{sj})G^{sl}\\
\!\!\!\!&=&\!\!\!\!L^s_{is}+B^s_{is}-L^s_{is}=B^s_{is}.
\end{eqnarray*}
Let $H^i_{jk}$ are coefficients of Cartan connection. Since $B^i_{jk}=H^i_{jk}-L^i_{jk}$, then we get
$div(X^H)=X^idiv(\delta_i)=X^iH^s_{is}-X^iJ_i$. Then it is easy to check that $H^s_{is}=\frac{1}{\sqrt{g}}\delta_i(\sqrt{g})=\delta_i(\ln\sqrt{g})$.
\end{proof}
\begin{cor}
Let $(M,K)$ be a Cartan space with Berwald connection. Then $div(X)=0$ if and only if $J_i=\delta_i(\ln\sqrt{g})$.
\end{cor}
Let us  define $gradf$ by
\[
G(gradf,X)=Xf,\ \ \ \forall X\in\chi(T^{\ast}M).
\]
Then in the adapted frames $\{\delta_i,\dot{\partial}^i\}$, one can yields
\begin{eqnarray*}
&&G(grad,\delta_i)=\delta_if=\nabla_{\delta_i}f,\\
&&G(grad,\dot{\partial}^i)=\dot{\partial}^if=\nabla_{\dot{\partial}^i}f.
\end{eqnarray*}
Put $gradf:=\alpha^i\delta_i+\beta_i\dot{\partial}^i$. Then from the above equations we have
\[
\alpha^i=G^{ih}\nabla_{\delta_h}f,\ \ \
\beta_i=G_{ih}\nabla_{\dot{\partial}^h}f.
\]
Therefore we conclude  the following.
\begin{proposition}
Let $(M,K)$ be a Cartan space with Berwald connection. Then we have
\begin{equation}
gradf=G^{ih}(\nabla_{\delta_h}f)\delta_i+G_{ih}(\nabla_{\dot{\partial}^h}f)\dot{\partial}^i.
\end{equation}
\end{proposition}

\bigskip

The Laplace operator of a scalar field $f\in C^{\infty}(TM)$, is
then defined as
\begin{equation}
\Delta f=div(gradf).
\end{equation}
Then we have
\begin{theorem}\label{Laplace}
Let the Riemannian metric $G$ on $T^{\ast}M_0$ comes from a Cartan
space $(M,K)$. Then the Laplace operator has the
following form
\begin{equation}
\Delta
f=G^{ih}(\nabla_{\delta_h}f)(\frac{1}{\sqrt{g}}\delta_i(\sqrt{g})-J_i),\
\ \ \ f\in C^{\infty}(T^{\ast}M).
\end{equation}
\end{theorem}
By attention to Theorem \ref{Laplace}, we conclude that if $f$ be a
horizontally constant function, then $\Delta f=0$. We have, also
\begin{cor}\label{Cor1}
The Laplace operator $\Delta$ is vanish if and only if
$J_i=\frac{1}{\sqrt{g}}\delta_i(\sqrt{g})=\delta_i(\ln\sqrt{g})$.
\end{cor}
\begin{proposition}
Let the Riemannian metric $G$ on $TM$ is the K\"{a}hler metric with
components defined by (\ref{Kahler}), which is induced by the Cartan
structure $K$ on $M$. Then we have
\[
div({\textbf{C}}^{\ast})=0,\ \ \
div(\textbf{S})=p^i\delta_i(\ln\sqrt{g}),\ \ \ \Delta K^2=0,
\]
where $\textbf{S}=p^i\delta_i$ is the geodesic spray of $(M,K)$.
\end{proposition}
\begin{proof}
Since $p^iJ_i=0$ and $\nabla_{\delta_i}K^2=\delta_iK^2=0$, then by
using Proposition \ref{div} and Theorem \ref{Laplace}, the proof
will be complete.
\end{proof}
From above proposition we result that $div(\textbf{S})$ is zero if
and only if $\delta_i(\sqrt{g})=0$. Then by Corollary \ref{Cor1}, we have the following.
\begin{theorem}
Let $(M,K)$ be a Cartan space with Berwald connection. Suppose that Laplace operator is vanishes. Then  $div(\textbf{S})=0$ if and only if $K$ ia a mean Landsberg metric.
\end{theorem}

\noindent
Esmaeil Peyghan\\
Faculty  of Science, Department of Mathematics\\
Arak University\\
Arak,  Iran\\
Email: epeyghan@gmail.com

\bigskip

\noindent
Akbar Tayebi\\
Faculty  of Science, Department of Mathematics\\
University of Qom\\
Qom. Iran\\
Email: akbar.tayebi@gmail.com

\bigskip

\noindent
Ali Ahmadi\\
Faculty  of Science, Department of Mathematics\\
Arak University\\
Arak, Iran\\
Email: ali.ahmadi3078@gmail.com

\end{document}